%% file: blockchain_energy_conf_arxiv.tex
\documentclass[letterpaper, 10pt, conference]{IEEEtran}
\usepackage{hyperref}
\usepackage{amsmath,amsthm}
\usepackage{graphics, subfigure}

\usepackage{tikz}
\usepackage{mathtools}
\usetikzlibrary{shapes,arrows}
\usepackage{algorithm, algorithmic}
\usepackage[outdir=./]{epstopdf}
\usepackage{color}
\usepackage{multirow}
\usetikzlibrary{automata,positioning}
\usepackage{wrapfig}
\usepackage{rotating}
\IEEEoverridecommandlockouts

\usepackage{url}
\usepackage{bbm}

\usepackage{enumitem}
\newtheorem{theorem}{Theorem}
\newtheorem{proposition}{Proposition}
\newtheorem{definition}{Definition}
\newtheorem{assumption}{Assumption}
\begin{document}
\author{Arnob Ghosh, Vaneet Aggarwal, and Hong Wan\\Purdue University, West Lafayette IN}

\title{Exchange of Renewable Energy among Prosumers using Blockchain with Dynamic Pricing}
\date{}
\maketitle
\begin{abstract}
We consider users which may have renewable energy harvesting devices, or distributed generators. Such users can behave as consumer or producer (hence, we denote them as  prosumers) at different time instances. A prosumer may sell the energy to other prosumers in exchange of money. We consider a demand response model, where the price of conventional energy depends on the total demand of all the prosumers at a certain time.  A prosumer depending on its own utility has to select the amount of energy it wants to buy either from the grid or from other prosumers, or the amount of excess energy it wants to sell to other prosumers. However, the strategy, and the payoff of a prosumer inherently depends on the strategy of other prosumers as a prosumer can only buy if the other prosumers are willing to sell. We formulate the problem as a coupled constrained game, and seek to obtain the generalized Nash equilibrium. We show that the game is a concave potential game and show that there exists a unique generalized Nash equilibrium. We consider that a platform will set the price for distributed interchange of energy among the prosumers in order to minimize the consumption of the conventional energy. We propose a distributed algorithm where the platform sets a price to each prosumer, and then each prosumer at a certain time only optimizes its own payoff. The prosumer then updates the price depending on the supply and demand for each prosumer. We show that the algorithm converges to an optimal generalized Nash equilibrium. The distributed algorithm also provides an optimal price for the exchange market.
\end{abstract}
\input{Introduction_conf}

\section{System Model}
In this section, we will  describe the pricing behavior (Section~\ref{sec:price}), prosumer's utility functions (Section~\ref{sec:utility}), the constraints every prosumer has to satisfy (Section~\ref{sec:constraints}). Finally, we formulate the optimization problem for the  prosumer decision (Section~\ref{sec:formulation}). 
%
%
%

\subsection{Pricing structure}\label{sec:price}
We assume that the time is slotted. Each prosumer can act either as producer or consumer in a slot. Prosumers also decide how much to sell (producer) or buy (consumer) in the exchange market. The prosumers wants to decide  its demand over a certain time horizon $T$ (e.g., over a day, or over 8-9 hours period). Let the number of prosumers be $N$. Note that a prosumer does not need to have energy generation capability.

Let $l_{i,t}$ be the demand of prosumer $i$ to grid at during time $[t,t+1)$. The price selected by the ISO for the consumption of conventional energy is denoted as 
\begin{align}\label{eq:price}
p_t(\sum_{i}l_{i,t})=\gamma_t\sum_{i}l_{i,t}
\end{align}
where $l_{i,t}$ is the demand of the prosumer $i$ to the grid at time $t$. Note that $\sum_{i}l_{i,t}$ is the total load of the grid at time $t$. $\gamma_t$ is the parameter that is selected by the grid and depends on the time of the day.  How the grid should select $\gamma_t$ is out of scope for this paper and left for the future.  Note that setting a price which varies linearly with the total demand is quite common in the literature \cite{eksin,ECTA}. 

Note that the price of the grid not only depends on the load of the user $i$, but also on the loads of the other prosumers. The price is realized only when the load of all the prosumers is known. Thus, it leads to a game among the users, where the payoff inherently depends on the strategies chosen by the other users. A user is not aware of the exact demand of the other users, thus, it leads to an {\em incomplete information game}. Note that unlike the price which is independent of the total load, in the price strategy stated in (\ref{eq:price}), there will always be an upper limit of the load. This is because a user will try to resist consuming too much energy. {\em The price is the same to all buyers of conventional energy.}



\textbf{Exchange Price}: We assume that there is a platform which sets the prices for exchange among the prosumers. The prosumers select the amount of energy to be sold at different times in a blockchain. The prosumers also updates the amount of energy to be bought from other prosumers in a blockchain.  If the prosumer $i$ buys  energy from prosumer $j$ it pays an amount $\mu_{j,t}$ per unit of energy during time $[t,t+1)$. The platform obtains a fixed fee for every unit of energy exchanged. Hence, the platform will want to maximize the overall amount of energy to be exchanged among the prosumers. However, an incentive has to be provided to the prosumers to sell the excess energy to other prosumers. 

 \subsection{Prosumer's payoff}\label{sec:utility}
 We, now, compute the expression of each prosumer's payoff. 
 Let the prosumer $i$ buys $q_{i,j,t}$ amount of energy from prosumer $j$ during $[t,t+1)$. Let the prosumer $i$ sells an amount of $s_{i,t}$ during the time duration. Let the prosumer $i$'s demand (or, total consumption) during time duration $[t,t+1)$ be $d_{i,t}$. The user's utility $U_i$ only depends on the consumption $d_{i,t}$. 
 
 Let the prosumer $i$'s demand (or, total consumption) during time duration $[t,t+1)$ be $d_{i,t}$. The user's utility $U_i$ only depends on the consumption $d_{i,t}$. Thus, the user's profit of total payoff is 
 \begin{eqnarray}\label{obj}
 \sum_{t=1}^{T}\mu_{i,t}s_{i,t}-\sum_{j\neq i}\mu_{j,t}q_{i,j,t}-l_{i,t}\gamma_t(\sum_{i}l_{i,t})+U_i(d_{i,t})
 \end{eqnarray}
 The user's\footnote{In the sequel, we use the notation user and prosumer interchangeably.} payoff inherently depends on the decision of the other users. Note that the amount of energy a prosumer buys from the other prosumers also inherently depends on the amount of energy sold by the prosumer. 
 
 \begin{assumption}
 We assume that $U_i(\cdot)$ is concave, and continuous. 
 \end{assumption}
User's preference (or, comfort) increases with the consumption. However, the rate of increase of preference decreases with the demand $d_{i,t}$. The utility function is randomly drawn from a distribution. {\em The ISO or the other prosumers are unaware of the exact utility of a prosumer.} However, there may be a correlation among the prosumers. 

The utility may also have a temporal correlation. For example, a prosumer can only attain a utility if the total demand is satisfied within a deadline. The prosumer will have zero utility if the demand is unsatisfied within the deadline. Such kind of utility arises for the deferrable loads such as EV charging, or dishwasher. 

 \subsection{Problem Formulation: Constraints}\label{sec:constraints}
 In this subsection, we describe the system of constraints that each  prosumer has to satisfy while taking its own decision. 
 
 Some users may have a deferrable load such as EV charging, and dishwasher. The demand only needs to be fulfilled within a certain time. For example, the EV needs to be ready before 8 am (e.g., if the user is going to work). However, the individual load may vary over time. We denote the set of deferrable appliances as $\mathcal{A}_i$.  Suppose the load assigned to appliance $j$ of user $i$ is $x_{i,j,t}$ for the time duration $[t,t+1)$. Hence, we have 
 \begin{eqnarray}
 \sum_{t=1}^{T_j}x_{i,j,t}\geq X_j \forall i, \forall j\in\mathcal{A}_i\label{constr1}
 \end{eqnarray}
where $X_j$ is the amount of load required for appliance $j$.

Let the set of non-deferrable load of prosumer $i$ is assumed to be $\mathcal{B}_i$. Let $x_{i,j,t}$ be the load for appliance $j\in \mathcal{B}_i$. 
The user's total consumption during time $[t,t+1)$ is thus
\begin{eqnarray}\label{constr1b}
d_{i,t}=\sum_{j\in \mathcal{A}_i\cup\mathcal{B}_i}x_{i,j,t}  
\end{eqnarray}
The demand of prosumer $i$ has to exceed the demand for non-deferrable appliances at each instance. 
Thus, 
\begin{eqnarray}
d_{i,t,max}\geq d_{i,t}\geq d_{i,t,min}\label{constr2}
\end{eqnarray}
where $d_{i,t,max}$, and $d_{i,t,min}$ are known beforehand. 

Let the amount of energy sold to prosumer $j$ by prosumer $i$ in the exchange market during time $[t,t+1)$ be denoted as $s_{i,j,t}$. The total amount of load sold by prosumer $i$ $s_{i,t}$ must follow the following constraint. 
\begin{eqnarray}\label{constr3}
s_{i,t}=\sum_{j\neq i}s_{i,j,t}.
\end{eqnarray}

Similarly, the energy bought by prosumer $i$ from prosumer $j$ is denoted as $q_{i,j,t}$ during time $[t,t+1)$. Thus, 
\begin{eqnarray}\label{constr3a}
q_{i,t}=\sum_{j}q_{i,j,t}.
\end{eqnarray}
The energy bought from user $j$ must be smaller than the total energy sold by user $j$.  Also note that the sum of energies bought by all the prosumers have to be equal to the energy sold by prosumer $j$. Thus, we have
\begin{eqnarray}\label{constr3ab}
\sum_{i}q_{i,j,t}= s_{j,t} \forall j.
\end{eqnarray}
 Note that $s_{j,i,t}=q_{i,j,t}$. 
 
 Each prosumer has a renewable energy harvesting device which harvests $\bar{E}_{i,t}$ amount of energy \footnote{A prosumer may not have any renewable energy harvesting device. In that case, the renewable energy will be $0$.} during time $[t,t+1)$. The prosumer $i$ may also have a battery with capacity $B_{i,max}$. If the prosumer does not have any battery then $B_{i,max}$ is $0$. The state of the battery is $B^{i,t}$. The amount of energy discharged from the battery is $e_{i,t}$, and charged to the battery is $b_{i,t}$. Thus, we have
 \begin{eqnarray}
 B^{t+1}=B^{t}+\bar{E}_{i,t}-e_{i,t}+b_{i,t}\label{constr5}
 \end{eqnarray}
 Note that the renewable energy generation is a random process. Hence, a prosumer will only have the estimate of $\bar{E}_{i,t}$, rather than the exact realized value. 
 
The state of the battery can not be less than $0$.  The state of the battery is also required to be a specific value at the end of the horizon. Most often, the state of the battery is kept to be same as the start of the day. Thus, we have
 \begin{eqnarray}
 B^{T+1}=B^{1}, 0\leq B^{t}\leq B_{max}\label{constr6}
 \end{eqnarray}
 
 Note that, the energy bought from the grid as well as from the other prosumers also have a transmission loss. Let $r_{i,j}\leq 1$ be the transmission efficiency between the users $i$ and $j$, and $r_i$ be the transmission efficiency between the user $i$ and the grid. Thus, the total consumption of prosumer $i$ during time $[t,t+1)$ is given by
 \begin{eqnarray}\label{constr7}
 d_{i,t}=\sum_{j}q_{i,j,t}r_{i,j}-s_{i,t}+l_{i,t}r_i-b_{i,t}\eta_{d}+e_{i,t}\eta_{c},
 \end{eqnarray}
 where $\eta_{d}\leq 1$ and $\eta_{c}\leq 1$ are respectively the discharging, and charging efficiency from the battery. Note that only a portion of energy bought by the prosumer can be used because of the transmission loss.
 
 We assume that 
 \begin{assumption}\label{assum:trnseff}
 $r_{i,j}> r_{i}$ for all $i$, and $j$. 
 \end{assumption}
We consider a geographically co-located prosumers. The extension of our work for prosumers situated in a vast geographical area is left for the future. Generally, the transmission efficiency from a neighbor prosumer should be high compared to obtaining energy from the grid. This is because the grid often times obtain energy from a far greater distance than the distance between local neighbors. 

 \subsection{The formulated problem}\label{sec:formulation}
 We are now ready to specify the optimization problem that each prosumer solves for every time horizon $[t,t+T)$. 
 \begin{eqnarray}
& P_i:  \text{maximize } & \sum_{t=1}^{T}\mu_{i,t}s_{i,t}-\sum_{j\neq i}\mu_{j,t}q_{i,j,t}\nonumber\\ & &-l_{i,t}\gamma_t(\sum_{i}l_{i,t})+U_i(d_{i,t})\nonumber\\
& \text{subject to } & (\ref{constr1})-(\ref{constr7}).\nonumber\\
& \text{var} & d_{i,t}, l_{i,t}, q_{i,t}, s_{i,t}, e_{i,t}, \nonumber\\ & &b_{i,t},q_{i,j,t}, s_{i,j,t} \geq 0\label{constr8}.
\end{eqnarray}
The problem is convex when the prosumer knows others' strategies. However, a prosumer is unaware of the strategies set by other prosumers. 

We must have $q_{i,t}s_{i,t}=0$. Later, we propose an algorithm which sets the prices in a way such that $s_{i,t}q_{i,t}=0$. 

\subsection{Platform's Objective}
The platform wants to maximize the exchange of energy among the prosumers or minimize the conventional energy consumption. Thus, the objective of the platform is the following
\begin{align}\label{eq:plat}
& \text{minimize } & \sum_{t=1}^{T}\sum_{i}l_{i,t}\nonumber\\
& \text{var }:& \mu_{j,t} \geq 0\forall j
\end{align}
However, the platform does not know exact utility parameter, and the user specific parameters such as the renewable energy generation, and the battery capacity of the prosumers. Note that the platform's objective (cf.(\ref{eq:plat})) has to be minimized carefully. If the exchange price is small enough the prosumers will not sell. On the other hand, if the exchange price is large,  prosumers will not buy enough amount from the prosumers who have excess energy. Thus, selecting the optimal $\mu_{j,t}$ is challenging for the prosumers. In Section~\ref{sec:distributed} we show how the prosumers should select the prices which will optimize the above problem.

\section{Solution Methodology}
Each prosumer is a selfish entity. It is only entitled to maximize its own payoff. The prosumer's payoff inherently depends on the other prosumers' strategies. We, thus, formulate the problem of prosumer's decision as a game theoretic model. We show that the problem is coupled constrained game since the constraints of the prosumers are coupled. We seek to obtain  generalized Nash equilibrium. In Section~\ref{sec:potential} we show that the game admits a concave potential function, and, thus, there exists a unique  generalized Nash equilibrium. Leveraging on the concave potential game, we propose a distributed algorithm (Section~\ref{sec:distributed}) where each prosumer updates its strategy based on the decision taken by the other prosumers in the previous period. However, the equilibrium is {\em not efficient} as it does not solve the problem where the objective is to maximize the sum of the payoffs of the prosumers. We show that as the renewable energy generation increases, the efficiency increases. 

\subsection{Nash Equilibrium}
Each prosumer only wants to maximize its own payoff.  Each prosumer's optimal decision depends on other users' strategies. However, a user is not aware of the strategies of other users. Hence, a game theoretical model is most apt to find the strategy of a prosumer. 
\begin{definition}
Let $a_i$ be the strategy of player $i$, and $a_{-i}$ be the strategy vector of all players except player $i$. 
\end{definition}
In a game theoretic setting, Nash equilibrium is used as an equilibrium concept. Naturally, the question arises whether there exists a Nash equilibrium in the game. The Nash equilibrium is defined in the following
\begin{definition}
The strategy profile $(a_i,a_{-i})$ is a Nash equilibrium strategy profile if the following holds 
$E[u_i(a_i,a_{-i})]\geq E[u_i(a^{\prime}_i,a_{-i}]$ for all $i$ and $a^{\prime}_i\in S_i$ where $S_i$ is the set of strategies of the player $i$. 

If the strategy space of a player depends on the other players, the game is known as coupled constrained game. The generalized Nash equilibrium is the corresponding equilibrium concept in the coupled constrained game.
\end{definition}
Specifically, in a Nash equilibrium strategy profile, any player can not have higher payoff by deviating unilaterally from the prescribed strategy profile. 
In many games, a Nash equilibrium may not exist. 
Even if the Nash equilibrium exists, the question arises whether it is unique. If it is unique,  the question  is whether it is practically implementable. Later, we answer the above questions  in an affirmative sense by applying the theory of potential game. 

We denote the set of strategy of prosumer $i$ as $\mathbf{a}_{i,t}$, where $\mathbf{a}_{i,t}=\{d_{i,t},\mathbf{Q}_{i,t}, q_{i,t},s_{i,t},\mathbf{S}_{i,t},e_{i,t},b_{i,t},l_{i,t}\}$.
$\mathbf{Q}_{i,t}$ ($\mathbf{S}_{i,t}$, resp.)  consists of the component $q_{i,j,t}$ ($s_{i,j,t}$, resp.)  for all $j$. The set $\mathbf{a}_{i,t}$ has to be feasible.

Note that because of the constraints in (\ref{constr3ab}) the strategy space of a prosumer also depends on the strategy of other prosumers. Hence, the above game belongs to the coupled constrained game. We, next, compute the generalized Nash equilibrium using the concept of potential game. 
\subsection{Potential Game}\label{sec:potential}
We, now, show that the game is a potential game. First, we introduce the definition of the potential game--
\begin{definition}
Suppose user $i$'s payoff function is $u_{i}(a_i,a_{-i})$ where $a_i$ is the strategy of user $i$, and $a_{-i}$ denotes strategies of other users apart from user $i$. Then a game is potential, if and only if there exists a function $\Phi$ such that $u_i(a^{\prime}_i,a_{-i})-u_i(a_i,a_{-i})=\Phi(a_i^{\prime},a_{-i})-\Phi(a_i,a_{-i})$, $\forall i$. 

If $\Phi(\cdot)$ is concave (strictly) then the game is (strictly) concave potential game \cite{potential}. 
\end{definition}
In the following, we show that the game we introduced is a potential game. 

\begin{theorem}
The game is potential. If $U_{i}(\cdot)$ is strictly concave $\forall i$, then it is a strictly concave potential game, and thus, it has a unique pure generalized Nash equilibrium. 
\end{theorem}
\begin{proof}
Let us assume that  $u_i(a_i,a_{-i})=\sum_{t=1}^{T}\mu_{i,t}s_{i,t}+U_i(d_{i,t})-\sum_{j\neq i}\mu_{j,t}q_{i,j,t}-l_{i,t}\gamma_t(\sum_{i}l_{i,t})$ where $a_i$ is itself a vector consisting of $d_{i,t}, l_{i,t}, q_{i,t}, s_{i,t}, e_t,b_t$. Now, consider the following potential function
\begin{align}
\Phi(a_i,a_{-i})=& \sum_{i}(\sum_{t=1}^{T}\mu_is_{i,t}+U_i(d_{i,t})-\sum_{j\neq i}\mu_jq_{i,j,t}+U_i(d_{i,t}))\nonumber\\-& \sum_{t=1}^{T}\gamma_t\sum_{i,j}l_{i,t}l_{j,t}
\end{align}
$\Phi(\cdot)$ is a potential function since
\begin{align}
u_i(a^{\prime},a_{-i})-u_i(a_i,a_{-i})=\Phi(a^{\prime},a_{-i})-\Phi(a_i,a_{-i}).
\end{align}

$\Phi(\cdot)$ is concave since $U_i(\cdot)$ is concave.
\end{proof}

Thus, the equilibrium of the game is given by the optimal solution of the potential function subject to the constraints. 

We, now, state the optimization problem, the solution of which will give the generalized Nash equilibrium.
\begin{eqnarray}\label{eq:pot}
P_{\mathrm{potential}}:  \text{maximize } & \sum_{t=1}^{T}\mu_{i,t}s_{i,t}+U_i(d_{i,t})\nonumber\\ & -\sum_{j\neq i}\mu_{j,t}q_{i,j,t}-l_{i,t}\gamma_t(\sum_{i}l_{i,t})\nonumber\\
\text{subject to } & (\ref{constr1})-(\ref{constr8}).
\end{eqnarray}
Since the potential game is concave,  the solution of the problem can be obtained using the convex optimization tool. 

\subsection{Efficiency of the equilibrium}
We, now, investigate the efficiency of the generalized Nash equilibrium obtained in the last section. Our result shows that the game is not efficient. {\em Specifically, the equilibrium strategy profile of the game is not optimal for maximizing the sum of the profits of the users.}In the following we compute the efficiency gap.
Let
\begin{align}\label{eq:p}
P: \text{maximize } & \sum_{t=1}^{T}\mu_{i,t}s_{i,t}+U_i(d_{i,t})-\sum_{j\neq i}\mu_{j,t}q_{i,j,t}\nonumber\\ &-l_{i,t}\gamma_t(\sum_{i}l_{i,t}) -\gamma_t(\sum_{i}l_{i,t})^2\nonumber\\
\text{subject to } & (\ref{constr1})-(\ref{constr8}).
\end{align}
Let $\mathbf{a}_{i,t}=(l_{i,t},s_{i,t},q_{i,t},b_{i,t},e_{i,t},d_{i,t})$ be the decision vector for each user $i$ at time $t$. The problem $P$ is a convex optimization problem as the objective function is concave, and the constraints are linear. 

We, next, introduce some notations. 
\begin{definition}
Let $\mathbf{a}_{i,t}^{*}$ be the optimal solution of the problem $P$. Let $\mathbf{a}_{i,t}^{eq}$ be the equilibrium strategy profile. Specifically, it is the optimal solution of $P_{\mathrm{potential}}$.
\end{definition}
Now, we are ready to characterize the efficiency ratio of the equilibrium strategy profile. 
\begin{definition}
Let the objective value attained by the strategy profile $\mathbf{a}^{eq}_{i,t}$ in the problem P be $P^{eq}$, and let $P^{*}$ be the optimal value. Then, the efficiency is 
\begin{align}
\eta=\dfrac{P^{eq}}{P^{*}}.
\end{align}
\end{definition}

The efficiency is always less than or equal to $1$. Note that in the equilibrium strategy profile, the user consumes larger energy as compared to the optimal solution from the grid because of the additional term. 

\begin{proposition}
The efficiency decreases as the number of prosumers increases. 

The efficiency $\eta$ increases as the amount of renewable energy increases. 
\end{proposition}
The above proposition shows that the efficiency increases when the renewable energy supply increases. Thus, as the renewable energy penetration increases, the equilibrium becomes closer to the optimal solution of the sum of the prosumers' utility maximization problem. 

If the load  to the grid increases, the efficiency decreases.  However, the efficiency is never $0$. It is always upper bounded above $0$. However, as the number of user increases the the efficiency decreases. 

\subsection{Distributed Solution}\label{sec:distributed}
Though $P_{\mathrm{potential}}$ is a convex optimization problem, one needs to know each prosumer's utility, and constraints. Thus, a centralized solution is difficult to obtain in practice. In this section, we show that how a distributed algorithm converges to the generalized Nash equilibrium. 

Each user updates its strategy for a certain time horizon $T$, which we denote as epoch.  At epoch $k$, the prosumer decides for the time slots $kT+1,\ldots, (k+1)T$. The platform initially selects exchange prices $\mu_{j,t}$ for each prosumer $j$. The prosumer then updates its strategy profile for each time slot in the epoch. The platform then again updates the price until the process converge. The algorithm \textbf{ALGO-DIST} is  detailed in the following 

\textbf{ALGO-DIST}:\\
\begin{enumerate}
\item Initialization: For each prosumer $i$, the price is set at $\mu_{i,t}^{1}$ at a some minimum possible value. 
\item At iteration $k\geq 1$, each prosumer $i=1,\ldots N$, updates its strategy by obtaining $\mathbf{a}^{k}_{i,t}=(l_{i,t}^{k},d_{i,t}^{k}, q_{i,t}^{k},s_{i,t}^{k}, e_{i,t}^{k}, b_{i,t}^{k})$, while solving the following problem
\begin{align}
 \text{maximize } & \sum_{t=1}^{T}\mu_{i,t}^{k}s_{i,t}+U_i(d_{i,t})-\sum_{j\neq i}\mu_{j,t}^{k}q_{i,j,t}-l_{i,t}\gamma_tl_{i,t}\nonumber\\  & -\gamma_t\sum_{j\neq i}l^{k-1}_{i,t}U_i(d_{i,t}) -1/\alpha_k\sum_{t=1}^{T}||\mathbf{a}_{i,t}-\mathbf{a}^{k-1}_{i,t}||^2\nonumber\\
\text{subject to } & (\ref{constr1}), (\ref{constr2}), (\ref{constr3}), (\ref{constr3a}), (\ref{constr5}), (\ref{constr6}), (\ref{constr7}), (\ref{constr8}).\nonumber
\end{align}
\item If $s_{i,t}<\sum_{j\neq i}q_{j,i,t}$, set $\mu_{i,t}^{k+1}=\mu_{i,t}^{k}+\epsilon$, if $s_{i,t}\geq \sum_{j\neq i}q_{j,i,t}$, set $\mu_{i,t}^{k+1}=\mu_{i,t}^{k}$,  and go to step $2$.
\item If $s_{i,t}\geq \sum_{j\neq i} q_{j,i,t}$ for all $i$, then exit.
\item The user $i$ pays $l^{k}_{i,t}\gamma_{t}\sum_{j}l^{k}_{j,t}$ at time $t$ to the grid, and obtains the value of other users' load. 
\end{enumerate}
Here, $\alpha_k$ is the learning parameter, and it is set at $1/(k+1)$. The parameter puts a weight $\alpha_k$ on the optimal decision of the current stage. The above also ensures that the optimal decision at the current iteration is not very far away from the decision of the previous iteration. 

Note that a prosumer does not have constraint in (\ref{constr3ab}). Thus, the prosumer is using a relaxed version of the problem since the prosumer does not know the demand of a prosumer to other prosumers. The price is set according to the total supply and demand. In the converged solution, the constraint in (\ref{constr3ab}) is always satisfied. 

The following shows that the above algorithm converges to an equilibrium. 
\begin{theorem}
The distributed algorithm \textbf{ALGO-DIST} terminates after a finite number of iterations. The solution converges to the optimal solution of the potential game for small enough $\epsilon>0$. 

The price $\mu_{i,t}$ converges to the exchange price such that the constraint in (\ref{constr3ab}) is satisfied, and the platform's objective (cf.(\ref{eq:plat})) is minimized. 
\end{theorem}

It is apparent that as $\epsilon$ becomes smaller, the algorithm converges to the solution of the potential game. Specifically, the algorithm converges to the generalized Nash equilibrium.

Note that in the algorithm, each prosumer only needs to know the load  of the other users during the  time slots  $(k-1)T+1,\ldots, kT$ for deciding its own decision. The prosumer can obtain the information from the blockchain where each prosumer updates its strategy.  It does not need to know the utility functions of the other users, their energy surplus, or even the battery capacities. Thus, the algorithm is easy to implement, and yet, it converges to the generalized Nash equilibrium. 

The user uses the data over the last epoch to obtain its own decision in the current epoch. Surprisingly, a simple algorithm converges to an equilibrium strategy. The above type of strategy where a prosumer learns about the strategy of the other prosumers, is known as {\em fictitious play}.  Since the potential function is strictly concave, the convergence is exponential \cite{fictitious}. 

The platform updates the price for each prosumer depending on the supply and demand for each prosumer. If the supply exceeds demand for a prosumer, it will decrease the corresponding price. On the other hand, if the demand is higher than the supply, it will increase the price. 

  Note that each prosumer knows the value of  $\gamma_t$ while updating its decision. $\gamma_t$ is set by the ISO, and changes in a longer time scale. Since the estimate of the renewable energy generation changes in real time, the prosumer can update its strategy at a minute scale. The convergence is exponential, thus, the algorithm can be adopted dynamically to obtain the prices. The initial price can be set at the lowest possible price for the conventional energy, which will ensure that the prosumers always try to buy from the exchange market. 
  
  {\em The exchange prices are dynamic, and is different for each prosumer}. $\mu_{i,t}$ attained by Algorithm \textbf{ALGO-DIST} is the maximum possible selling price in the exchange market which will minimize the conventional energy consumption. It also depends on the excess energy a prosumer has.  During the peak time, the conventional energy price is likely to be more. The prosumers will be  more likely to sell energies during that time. The above has threefold advantages. First, the prosumers will earn more. Second, the peak conventional energy demand will also reduce. Finally, the price for the conventional energy will also go down.

\section{Numerical Analysis}\label{sec:simulation}
In this section, we show the equilibrium load profile, the reduction of the conventional energy consumption, and the efficiency of the generalized Nash equilibrium in an example setting. 
\subsection{Parameter Setup}
The conventional energy generation incurs a cost to the grid. We assume that  the cost function $C(\cdot)$ is a piece-wise quadratic function
\begin{align}
C(L)=\begin{cases} a_1L^2+c_1,\quad \text{if } L\leq L_1\nonumber\\
a_2L^2+c_2, \quad \text{else }\end{cases}
\end{align}
For the simulation $L_1$ is assumed to be $2$MW. $a_1$, $c_1$ are taken as $1$. $a_2$, $c_2$ are taken as $2$. 
The utility function of the prosumers is defined in the following manner 
\begin{align}
U_i(d)=\beta_t d- \zeta d^2.
\end{align}
where $\beta_t$ is a time-dependent random parameter which has a Gaussian distribution of mean $0.3$\$/kW during the off-peak period (for $t\in [6,9], t\in [13,15]$), and has a Gaussian distribution of mean $0.6$\$/kW during the peak period (for $t\in [9,12]$ and $t[16,18]$). $\beta_t$ has a variance of $30$ cents for all $t$. $\zeta$ is assumed to be $0.1$ \$/(kW-h)$^2$. We assume that $\gamma_t$ is the same across the day. More specifically, we set $\gamma_t$ at $0.15$ for all $t$ apart from Fig.~\ref{fig:gammvssoc}. 

We assume that the renewable energy is one-sided normally distributed with mean $2$ kW and variance of $1$ (kW)$^2$. We assume that the battery capacity is assumed to be $10$kWh.  We use the Algorithm \textbf{ALGO-DIST} with $\epsilon=10^{-6}$ to obtain the solution for our approach.

\begin{figure}
\centering
\includegraphics[width=.33\textwidth]{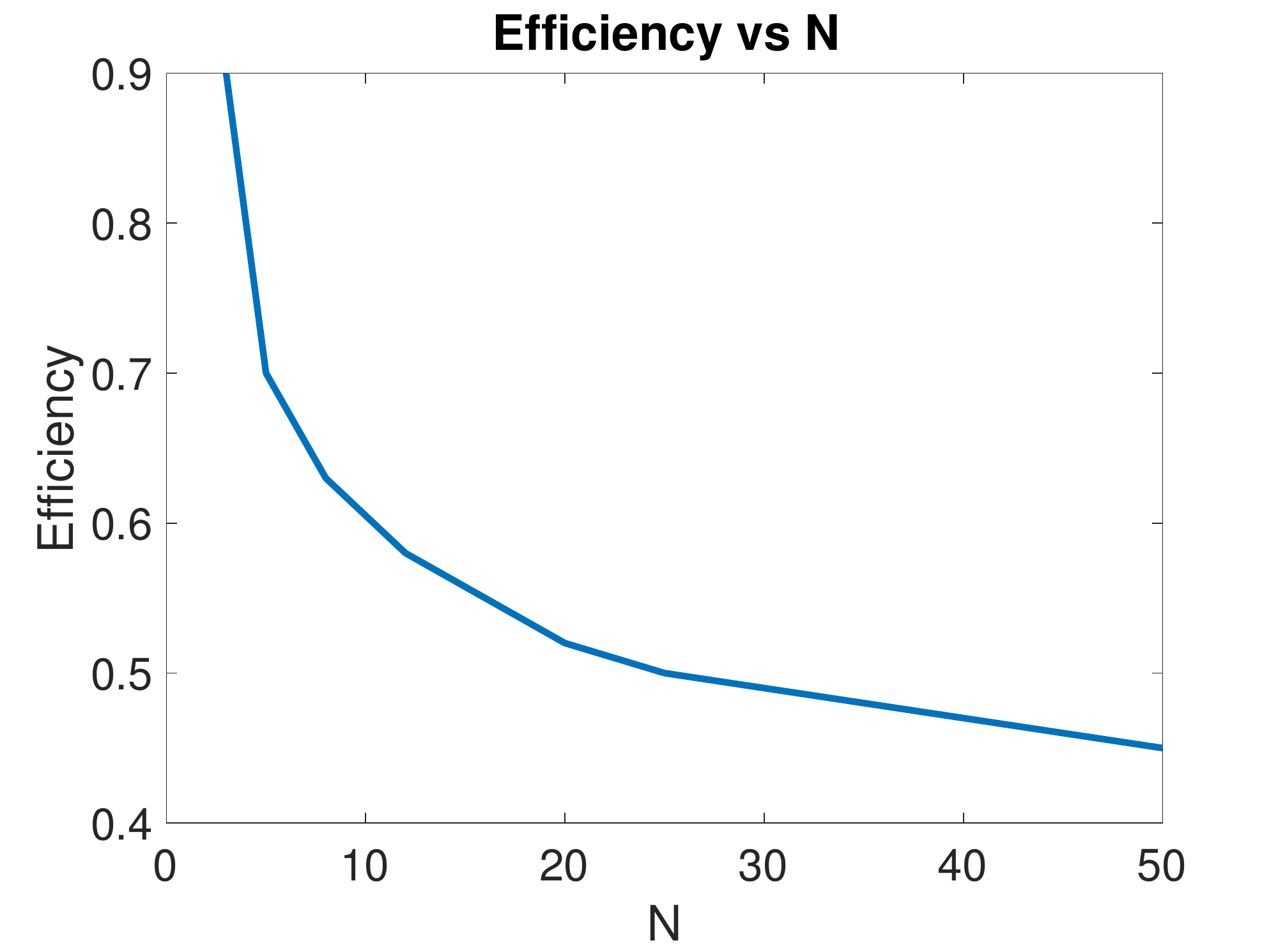}
\vspace{-.2in}
\caption{The variation of efficiency of user's equilibrium strategy profile.}
\label{fig:eff}
\vspace{-.2in}
\end{figure}

The transmission loss is assumed to be symmetric and equal for each prosumer. Specifically, $r_{i,j}=0.9$ for all $i$ and $j$. The transmission loss between the grid and prosumers is also assumed to be the same with $r_i=0.8$ for all $i$. 
\subsection{Results}
\subsubsection{Efficiency with number of prosumers}
Our first result (Fig.~\ref{fig:eff}) shows the efficiency of the equilibrium strategy profile with the number of prosumers. As the number of prosumers increases, the efficiency decreases, however, the decrement slows down when the number of prosumers exceeds a certain threshold. The variation is non-linear and there is a lower bound as shown in Theorem 2. 

\subsubsection{Impact of exchange market on the peak load}
Fig.~\ref{fig:eta} shows that the exchange market reduces the peak demand by at least 30\%. When there is no exchange market, we assume that the prosumers solve the problem with $\mu_{i,t}$ is set at a high value for all $i$ and $t$ in the optimization problem $P$ (cf.~(\ref{eq:p})). Thus, the prosumers can only sell excess energies to the grid, or use it at later time by storing in their batteries.  However, even the prosumers have the storage device, the exchange market can greatly enhance the reduction of the peak load. This is because the prosumers can sell energies to the other prosumers in the peak time which reduces the peak load. The exchange prices are also high during the peak time as the demand to the prosumers increases during the peak time. 
\begin{figure}
	\centering
\includegraphics[width=.33\textwidth]{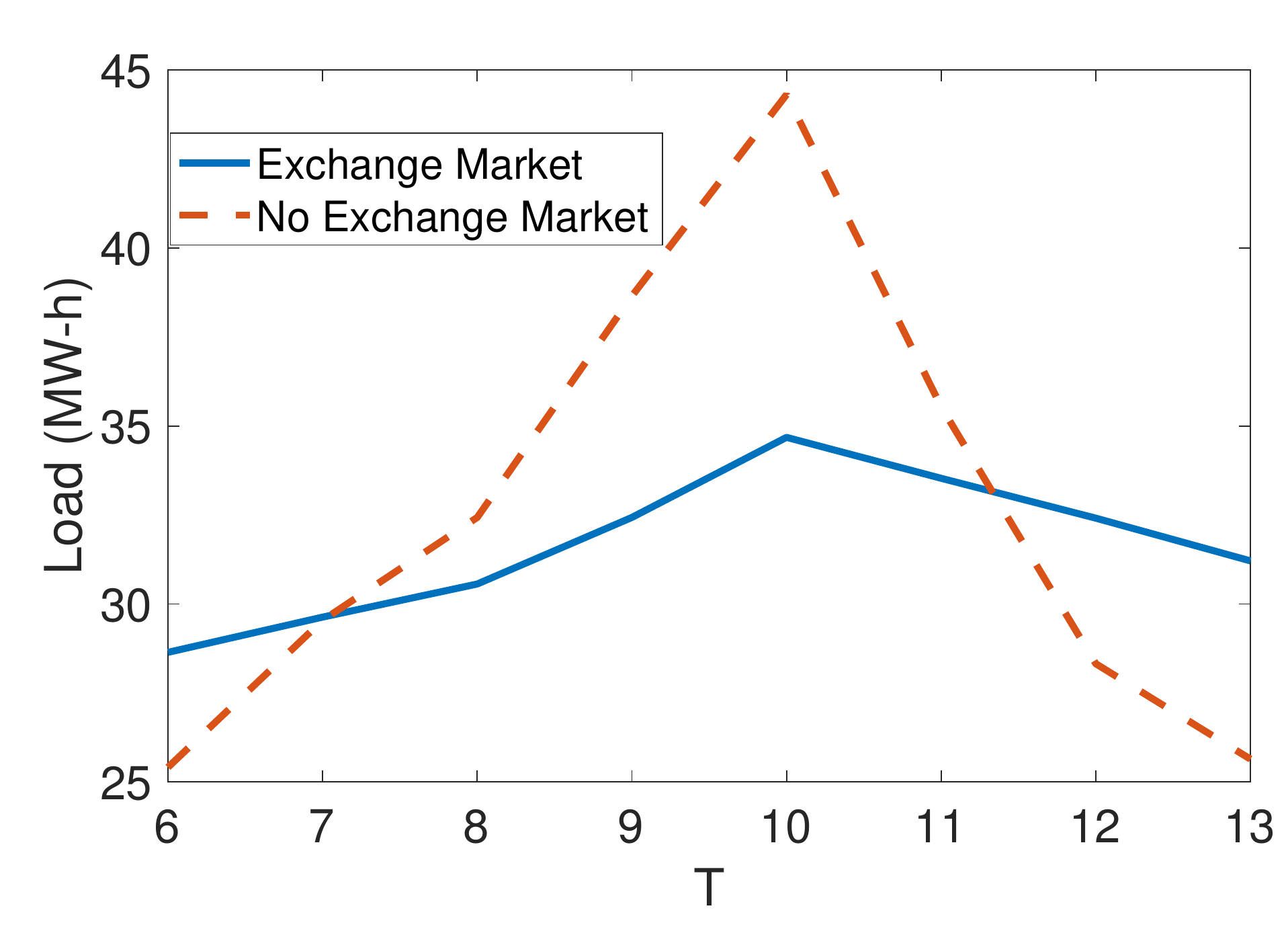}
\vspace{-.2in}
\caption{Variation of the demand of conventional energy with $T$ with exchange market, and no exchange market for $N=50$. The time $[9,12]$ is peak period, rest is off-peak period.}
\label{fig:eta}
\vspace{-.2in}
\end{figure}

\subsubsection{Increase of social welfare}
The social welfare is defined to be the difference between the utility function of the prosumers and the cost function of the conventional energy. A higher social welfare is always beneficial. In fact, independent system operator (ISO) wants to maximize the social welfare. 

Fig.~\ref{fig:exchange} shows the importance of the exchange market as compared to the one with no exchange market. Fig.~\ref{fig:exchange} shows that the social welfare is at least 25\% higher when there is no exchange market. We assume that when there is no exchange market the prosumers optimize with $\mu_{i,t}$ is set at a high value in optimization problem $P$ (cf.~(\ref{eq:p})) which precludes any exchange among the prosumers. The reduction of the social welfare is because of two factors. First, the transmission efficiency is  poorer between grid and the prosumers. Second, the price paid by the prosumers is much higher when there is no exchange market. As the storage capacity increases, the difference is more significant. However, when the storage capacity increases beyond a certain threshold, the difference decreases. This is because for higher storage, the prosumers become self sufficient. 

\begin{figure}
	\centering
\includegraphics[width=.33\textwidth]{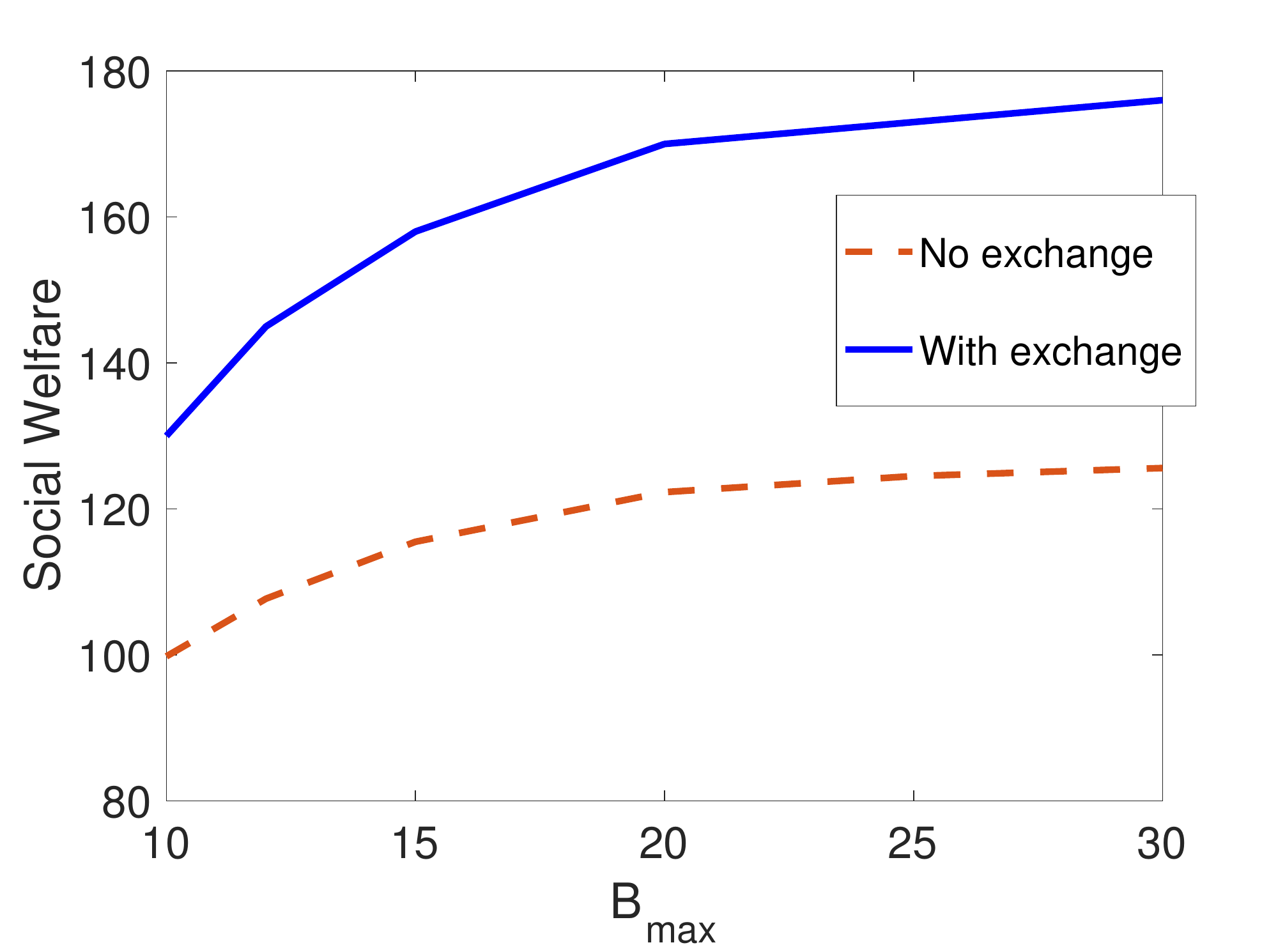}
\vspace{-.2in}
\caption{Variation of social welfare when there is an exchange market and when there is no exchange market with the battery capacity. }
\label{fig:exchange}
\vspace{-.2in}
\end{figure}

\begin{figure}
\centering
\includegraphics[width=.3\textwidth]{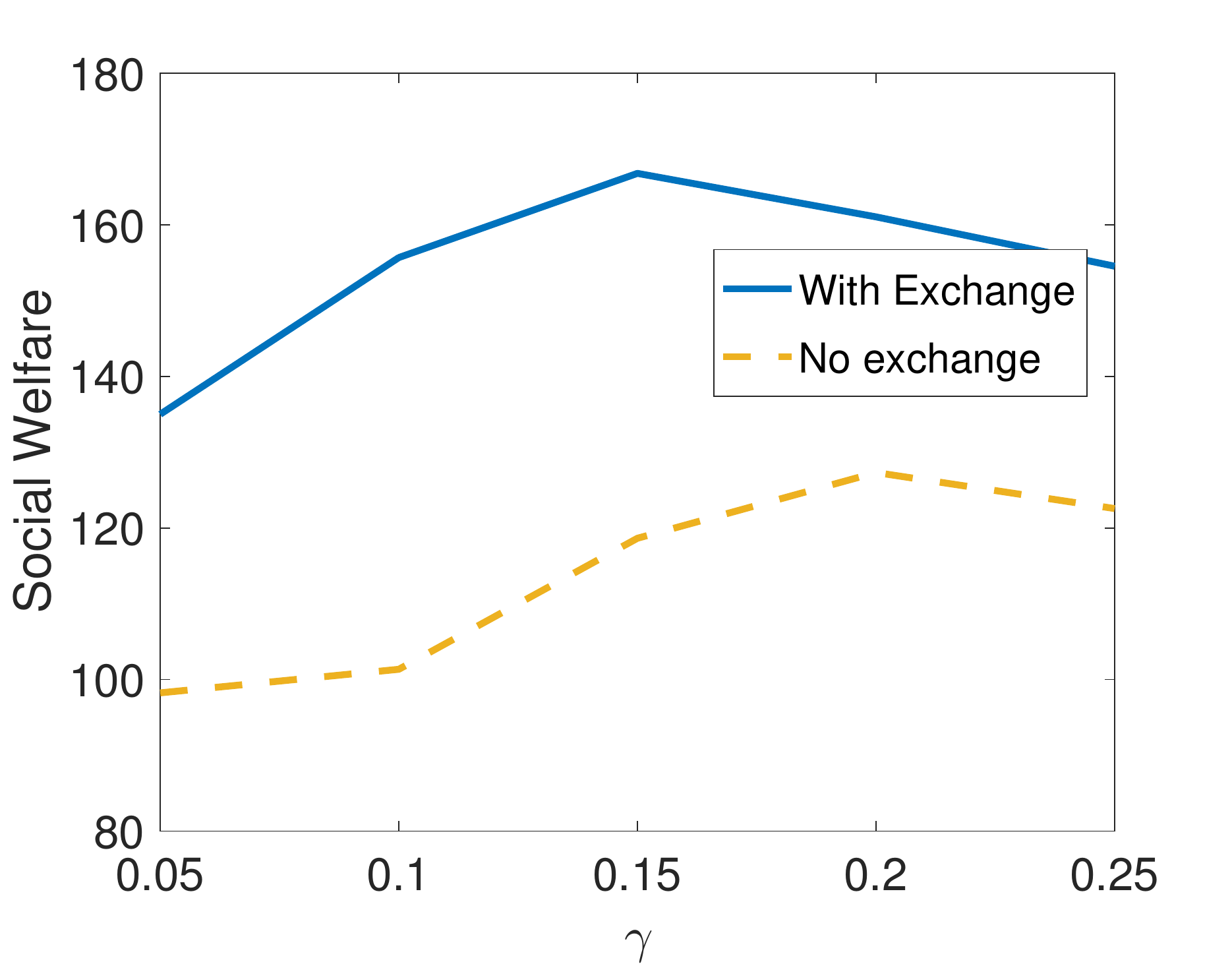}
\vspace{-.2in}
\caption{Variation of the social welfare when there is an exchange market and when there is no exchange market with $\gamma_t=\gamma$ for all $t$, and $N=50$.}
\label{fig:gammvssoc}
\vspace{-.2in}
\end{figure}

\subsubsection{Effect of $\gamma_t$}
Note that the ISO should select $\gamma_t$ in order to maximize the overall social welfare. Fig.~\ref{fig:gammvssoc} shows that when $\gamma_t$ is small, the prosumers' utility is higher as the prosumers have to pay a lower payment. The cost of conventional energy is higher as the demand of conventional energy will be higher in the above scenario. However, as $\gamma_t$ exceeds a threshold the social welfare decreases because the prosumers have to pay a higher amount which decreases the utility. The social welfare is higher when there is exchange market compared to the scenario when there is no exchange market for every value of $\gamma_t$. 
  \section{Conclusions and Future Work}
  We consider a system of prosumers which can share their excess energies among themselves in addition to buying conventional energy from the market. We consider a linear dynamic pricing for the conventional energy where the price for the conventional energy depends on the total demand for the conventional energy. The optimal strategy of a prosumer depends on the strategies of other prosumers. We, thus, model the strategy selection problem of a prosumer as a game theoretic problem. The strategy space of a prosumer inherently depends on the strategy of other prosumers since the amount of energy a prosumer can buy inherently depends on the amount of energy other prosumers want to sell. We, thus, seek to obtain a generalized Nash equilibrium. We show that the game admits a concave potential function. We propose a distributed algorithm where each prosumer only selects its own optima strategy. The platform selects the exchange price for each prosumer. A prosumer does not need to know the other prosumer's utilities. It only needs to know the strategy taken by a prosumer in the previous iteration to find its own in the current iteration. The platform updates the price depending on the supply and demand. We show that the distributed algorithm converges to the unique generalized Nash equilibrium, and the optimal exchange price. In the optimal exchange price, total conventional energy consumption is reduced, the peak load is reduced. 
  
 We do not consider how the ISO should set the parameter $\gamma_t$. The characterization of the optimal $\gamma_t$ is left for the future. We consider a real time price which varies linearly with the total demand. The characterization of the generalized Nash equilibrium for other non-linear function, and its impact on the efficiency ratio also constitutes an interesting research direction. The impact of the uncertainty of the renewable energy on the equilibrium strategy profile is also left for the future.

%
%
\end{document}

%% file: Introduction_conf.tex
\section{Introduction}

The problems of the peak demand, and the uncertainty of the demand and supply are challenging the traditional power grid. The integration of the renewable energy into the system has increased the uncertainty because of the randomness of renewable energy generation. The ever increasing use of the electric vehicles have also increased the demand of the users. Demand response mechanism has been proposed where time-of-use price is primarily used to shift the peak load during the off-peak times. However, those approaches assume that the users are {\em price takers} and do not consider the strategy of the other users. A consumer may consume a higher load, however the marginal price for each increment is the same. Those prices are set with an estimate of a users' demands. However, in the real time, the user's demand may be quite different from the estimate, and thus, the cost or the peak load may still be large. 

The continuing proliferation of the distributed energy resources (e.g., PV arrays, solar rooftops, energy storage units) have transformed the notion of traditional users of energy. The consumers can now also produce, and reduce the burden of the grid. Other users can buy the energy directly from those intermittent sellers. We denote the consumers with the capability of producing energies as prosumers. The prosumers can exchange necessary energies among themselves. The proliferation of the secure distributed database such as blockchain can facilitate such an exchange. Such local exchange is also useful as the transmission loss will be lower as compared to the scenario where the grid has to serve the users. However, a proper incentive is required for prosumers to participate for exchanging energy with others. 

Thus, a proper pricing mechanism is required which will facilitate the exchange of energies among the prosumers, and will decrease demand of the conventional energy from the grid. Further, the pricing mechanism has to be computationally simple, because the exchange has to be taken place in real time.  The determination of a proper exchange price is challenging. This is because of the uncertain nature of the renewable energy. Further, the utilities of each prosumer are also unknown. Finally, the exchange price should be set with an objective to reduce the overall cost or the peak demand of the grid. Hence, the exchange between the prosumers will not depend on the price at a certain time but also on the exchange prices at other times. The exchange will also depend on the prices of conventional energy which has to be bought from the grid. This paper seeks to consider these interactions to propose a pricing mechanism for prosumers where the prosumers can exchange among themselves. 

We consider a scenario where the prosumers can sell energies to other prosumers. The time is slotted. Each prosumer has certain demand to be fulfilled. Some demands are deferrable, i.e., the prosumer only cares about the total demand to be fulfilled over  multiple time slots. Each prosumer may also have a renewable energy harvesting device, and a storage device. Given the exchange price, the prices of the conventional energy, and the renewable energy prediction, each prosumer decides how much to buy from the conventional energy from the grid, how much to store, how much to sell (if any) to other prosumers, and how much to buy from other prosumers at each instance. {\em The prosumers are selfish entities which only want to maximize their own payoffs.}

We consider that there exists a platform which sets the exchange price. Examples of platform may be the load serving entity (LSE), or the aggregator, or a private retailer\footnote{Recently, there are private retailers which are providing exchange service among the distributed generators\cite{distr_retail}.}. The platform wants to maximize the amount of exchange such that the overall consumption of the conventional energy is reduced. However, as mentioned before, finding the optimal pricing strategy is inherently challenging because the platform does not know the exact utilities of the prosumers.  

We consider a real-time price for the conventional energy. The Independent system operator (ISO) sets a price which is a linear function of the consumption of all the users. The price is realized when the loads of all the users are provided to the ISO. Thus, each prosumer has to consider the strategy of other prosumers since the price inherently depends on the actions of other prosumers. The key characteristic of the pricing strategy is that the price is set in a dynamic manner, and depends on the total consumption of the conventional energies across the users. We, thus, focus on finding the exchange price which will reduce the overall consumption of the peak load. 

Since the prosumers are selfish entities, we use a game theoretic model to characterize the interaction among the prosumers.  Each prosumer selects its best response by anticipating the behavior of others. However, the strategy space is also dependent on the strategy of the other prosumers. For example, the maximum energy bought by a prosumer ($A$, say) from a prosumer ($B$, say) is limited by the amount the prosumer $B$ wants to sell at a certain time. Different prosumers may not sell all the excess energy, rather, they can store the energy and sell it only during the peak period. Thus, finding an equilibrium strategy is inherently challenging. We resort to the concept of generalized Nash equilibrium. 

We show that the game is a potential game. In fact, the potential game is strictly concave when the utility functions of the users is strictly concave, and thus, admits a unique generalized Nash equilibrium. 
However, the above equilibrium is not an optimal solution of the scenario where the sum of the prosumers' utilities are maximized. In other words, the equilibrium is not an efficient one, however, if the renewable energy generation is high, the efficiency is also high. 

Subsequently, we propose a distributed algorithm to find the equilibrium, and the optimal price the platform should set for the exchange of energy. In the algorithm, the platform first sets a low price for each prosumer. Each prosumer then, decides the total load, the amount of energy to be bought from other prosumers, and the amount of energy to be sold.  While taking its decision, a prosumer takes the strategy of the previous iteration of other prosumers as an estimate of the current strategy of the prosumers. Thus, a prosumer does not need to know the utilities of other prosumers. The platform then increases the exchange price  by an $\epsilon$ amount for those prosumers for which the supply is less than the demand. The process continues till the supply of all the prosumer all exceeds the demand.

We show that such a simple distributed algorithm converges to the equilibrium strategy of the prosumers using the concept of the {\em fictitious play}. The exchange price also converges to the optimal price which maximizes the overall exchange of energy among the prosumers. The convergence is also in the exponential order. We, empirically, show that such a pricing strategy reduces the conventional energy consumption in Section~\ref{sec:simulation}.  


\section{Related Literature}
Demand response pricing has already been studied  \cite{albadi,li_low,rahimi,low3,aalami,nguyen}. The analysis of all  these papers is based on the assumption that consumers are price takers. They find the optimal price based on the estimated demand profile of the consumers. However, in real time, the consumers may consume more as compared to the estimate. The consumers will still pay a marginal price according to the demand which has been estimated a priori. We consider a real time pricing mechanism where the price is set based on the real time total demand. Thus, the users (prosumers) have to strategize their demand based on the anticipation of the strategy of the other users. The above papers also did not consider the scenario where the users can sell energies among themselves in a separate market setting. 

Real time pricing has been considered recently  \cite{chen,yoon,maharajan,yu,conejo}. \cite{eksin} proposed a game theoretic model where the consumers face a price which varies depending on the total load to the grid. However, the above papers did not consider the exchange market which is considered in our work. The game is now a coupled constrained game, where the constraints of each user (prosumer) also depends on the other prosumers' decisions since the a consumer can only know the amount a prosumer wants to share. Thus, an optimal exchange price is required which will enhance the exchange of energy among the prosumers. Finally, \cite{eksin} did not consider the temporal correlation of the demand of the users. The utility of the users, and the demand of the users often have temporal correlation. For example, in the EV charging, the users may need a certain amount of charged battery before a certain deadline. The consumption of conventional energy, and the exchange of energy thus, will also depend on the prices across a certain horizon. Finally, each prosumer may also have a storage device, and thus, may defer the selling of energy at a later time if the selling price is higher. 

Energy exchange among the users in a micro-grid setting has been considered \cite{lee_guo,saad,rad,saad2}. The papers considered that the buyers and sellers will exchange energy among each other by bidding asking price, and selling price respectively. However, the price of the conventional energy is assumed to be fixed and independent of the amount of energy exchanged among the buyers and sellers. However, the price for the conventional energy may also drop if the total demand of conventional energy is low. Further, the above papers assumed that the demand consists of non-deferrable loads. However, in practice users may defer the consumption of energy if it is better to exchange energy when the exchange price is high. We provide an equilibrium strategy of the users considering the real time pricing mechanism where the price increases based on the total consumption (rather than a fixed price), and we also provide pricing strategies to the platform for the exchange market. 
